\documentclass[11pt]{article}
\usepackage[margin=1in]{geometry}

\usepackage{amsmath,amsbsy}
\usepackage{amsthm} %
\usepackage{amssymb}
\usepackage{microtype} %
\usepackage{lmodern} %
\usepackage{setspace}
\usepackage{color}
\usepackage{xcolor}
\usepackage{algorithm}
\usepackage[noend]{algpseudocode}
\usepackage{multirow}
\usepackage{xspace}
\usepackage{threeparttable}
\usepackage{booktabs}
\usepackage{soulpos} %
\definecolor{MyGreen}{rgb}{0.1333,0.5451,0.1333}
\usepackage[linktocpage=true, 
	pagebackref=true,colorlinks,
  linkcolor=black,
	citecolor=MyGreen,
	bookmarks,bookmarksopen,bookmarksnumbered]
	{hyperref}
\usepackage{epsdice}

\usepackage{mathtools}
\DeclarePairedDelimiter\bra{\langle}{\rvert}
\DeclarePairedDelimiter\ket{\lvert}{\rangle}
\DeclarePairedDelimiterX\braket[2]{\langle}{\rangle}{#1\,\delimsize\vert\,\mathopen{}#2}

\usepackage{pifont}

\let\originalleft\left
\let\originalright\right
\renewcommand{\left}{\mathopen{}\mathclose\bgroup\originalleft}
\renewcommand{\right}{\aftergroup\egroup\originalright}

\newcommand{\set}[1]{\{ #1 \}}
\newcommand{\Set}[1]{\left\{ #1 \right\}}

\newcommand{\lt}{\left}
\newcommand{\rt}{\right}

\newcommand{\local}{\ensuremath{\mathsf{LOCAL}}}

\newcommand{\congest}{\ensuremath{\mathsf{CONGEST}}}

\makeatletter
\newtheorem*{rep@theorem}{\rep@title}
\newcommand{\newreptheorem}[2]{%
\newenvironment{rep#1}[1]{%
\def\rep@title{#2 \ref{##1}}%
\begin{rep@theorem}[restated]}%
{\end{rep@theorem}}}
\makeatother
\newreptheorem{lemma}{Lemma}
\newreptheorem{theorem}{Theorem}
\newreptheorem{corollary}{Corollary}

\newboolean{short}
\newcommand{\onlyShort}[1]{\ifthenelse{\boolean{short}}{#1}{}}
\newcommand{\onlyLong}[1]{\ifthenelse{\boolean{short}}{}{#1}}

\theoremstyle{plain}
\newtheorem{lemma}{Lemma}
\newtheorem{theorem}{Theorem}

\newtheorem{fact}{Fact} %
\newtheorem{open problem}{Open Problem}
\theoremstyle{remark}

\usepackage[most]{tcolorbox}
\tcbuselibrary{theorems}

\newcommand{\ann}[1]{%
\text{\footnotesize(#1)}\quad}

\usepackage{thmtools}
\usepackage{mathtools}
\usepackage[T1]{fontenc}

\title{The Quantum Message Complexity of Distributed Wake-Up with Advice} 

\date{}

\author{Peter Robinson\\
\small{School of Computer \& Cyber Sciences}\\
\small{Augusta University}
\and
Ming Ming Tan \\
\small{School of Computer \& Cyber Sciences}\\
\small{Augusta University}
}

\newcommand{\kt}{\ensuremath{\mathsf{KT}}}

\newcommand{\id}{\ensuremath{\text{id}}}

\newcommand{\arad}{\ensuremath{\rho_{\textsl{\tiny awk}}}}

\newcommand{\qsearch}{\mathsf{QuantumSearch}}
\newcommand{\iterqsearch}{\mathsf{IteratedQuantumSearch}}

\newcommand{\EReg}{\mathsf{QER}}
\newcommand{\ECReg}{\mathsf{CER}}
\newcommand{\RReg}{\mathsf{QRR}}
\newcommand{\RCReg}{\mathsf{CRR}}
\newcommand{\noneawake}{\textsf{NULL}}

\newcommand{\matching}{\mathsf{Matching}}
\newcommand{\desc}{\mathsf{Descriptor}}

\begin{document}
\maketitle
\thispagestyle{empty}
\pagestyle{empty}
\begin{abstract}
We consider the distributed wake-up problem with advice, where nodes are equipped with initial knowledge about the network at large. 
After the adversary awakens a subset of nodes, an oracle computes a bit string (``advice'') for each node, and the goal is to wake up all sleeping nodes efficiently.
We present the first upper and lower bounds on the message complexity for wake-up in the quantum routing model, introduced by Dufoulon, Magniez, and Pandurangan (PODC 2025).
In more detail, we give a distributed advising scheme that, assuming $\alpha$ bits of advice per node, wakes up all nodes with a message complexity of $O\lt( \sqrt{\frac{n^3}{2^{\max\{\lt\lfloor (\alpha-1)/2 \rt\rfloor,0\}}}}\cdot\log n \rt)$  with high probability. %
Our result breaks the $\Omega\lt( \frac{n^2}{2^\alpha} \rt)$ barrier known to hold for classical (i.e., non-quantum) algorithms in the port numbering model when considering sufficiently dense graphs. 

To complement our algorithm, we give a lower bound on the  message complexity for distributed quantum algorithms: %
By leveraging a lower bound result for the single-bit descriptor problem in the query complexity model, we show that wake-up has a quantum message complexity of $\Omega\lt( n^{3/2} \rt)$ without advice, which holds independently of how much time we allow.  
In the setting where an adversary decides which nodes start the algorithm, most graph problems of interest implicitly require solving wake-up, and thus the same lower bound also holds for other fundamental problems such as   single-source broadcast and spanning tree construction.
\end{abstract}
\pagestyle{plain}
\setcounter{page}{1}
\section{Introduction} \label{sec:intro}

We consider the fundamental wake-up problem in distributed message passing networks, where an adversary wakes up some arbitrary subset of nodes, and the goal is to execute a distributed algorithm that wakes up all sleeping nodes efficiently.
The wake-up problem has been studied in the ``distributed computing with advice'' setting, pioneered in \cite{fraigniaud2006oracle,fraigniaud2009distributed}, where an oracle first observes the entire network and then equips each node with a short bit string (called ``advice'').
It is easy to see that the wake-up problem can be solved with a time complexity linear in the network diameter if we allow nodes to use all communication links in every round, which is clearly optimal in the worst case. 
Similarly, it is not too difficult to observe that optimal time as well as optimal message complexity can be obtained if we allow the oracle to provide $\Theta\lt( \log n \rt)$ bits of advice per node (on average); see \cite{fraigniaud2006oracle}.
On the other hand, \cite{wakeup} showed that $\Omega\lt( n^{2} \rt)$ presents an insurmountable lower bound  on the message complexity in the (classical) port numbering model if nodes do not have access to advice.
Thus, the main challenge is designing wake-up algorithms that are message-efficient by sending across as few edges of the network as possible, while requiring only a small amount of advice per node.
Even though quantum advantages have been well established for the time complexity of distributed graph algorithms, e.g., see \cite{le2018sublinear,censor2022quantum}, and very recently also for locally-checkable labeling problems~\cite{balliu2026distributed}, for a long time it remained unclear  whether quantum capabilities can also lead to improved message complexity bounds.  

In a recent breakthrough, Dufoulon, Magniez, and Pandurangan~\cite{dufoulon2025quantum} introduced the quantum routing model and showed that there is indeed a \emph{quantum communication advantage} for distributed leader election, thereby obtaining a polynomial improvement over the existing lower bounds known for the classical setting~\cite{jacm15,kutten2015sublinear}. 
Even more recently, Le Gall, Luce, Marchand, and Roget~\cite{gall2025exponential} presented another instance of a quantum communication advantage in the same model, by demonstrating that even an exponential gap is possible for certain routing problems in networks. 
\cite{dufoulon2025quantum} and \cite{gall2025exponential}  assume that all nodes are awake and start executing the algorithm simultaneously, which means that we cannot directly use their algorithms in our setting. 

While existing algorithms achieve near-optimal complexity bounds for wake-up in the classical distributed setting (see \cite{wakeup,fraigniaud2006oracle}), in this work, we initiate the study of the wake-up problem in the distributed quantum routing model of~\cite{dufoulon2025quantum}, where nodes may communicate using qubits and perform quantum computation. 
We believe that our work takes the first step towards understanding the impact of initial knowledge (``advice'') on the message complexity of distributed quantum algorithms.

\subsection{Related Work} \label{sec:related}

\paragraph{The Wake-Up Problem and Adversarial Wake-Up.}
Assuming that an adversary decides which nodes start executing is a common assumption in works on distributed graph algorithms and leader election; e.g., see \cite{gallager1983distributed,afek1991time,singh1992leader,kutten2020singularly,dufoulon2022almost,wakeup}. 
The wake-up problem with advice was first introduced by Fraigniaud, Ilcinkas, and Pelc,  in \cite{fraigniaud2006oracle}, who showed that wake-up is possible with $O(n)$ (classical) messages if nodes receive $O(n \log n)$ bits of advice in total, summed over all nodes.
However, the worst case length of advice per node can be up to $O(n)$ in their setting.
In addition, \cite{fraigniaud2006oracle} proved that $\Theta\lt( n \log n \rt)$ bits of total advice is indeed a tight bound for achieving $O(n)$ message complexity with any deterministic algorithm.
More recently, \cite{wakeup} gave several advising schemes that trade-off messages, time, and advice length. 
In particular, they show how to obtain optimality in all three complexity measures up to logarithmic factors, and also give the first lower bound for the port numbering model that holds for randomized algorithms, and shows that any, {possibly randomized}, algorithm with an expected message complexity of $O(\frac{n^2}{2^{\alpha}\log n})$ must have a worst case advice length of $\Omega\lt(\alpha  \rt)$ bits per node.
We refer the reader to Table~1 of \cite{wakeup} for additional results on the wake-up problem in other settings such as the asynchronous model, the $\local$ model, and the $\kt_1$ assumption, where nodes start out knowing their neighbors' IDs.

Note that the wake-up problem considered here is different from the variant studied in radio networks (e.g., \cite{DBLP:conf/icdcn/BiswasY26,DBLP:journals/dc/DaniGHP23}) and the \emph{sleeping model}~\cite{DBLP:conf/podc/ChatterjeeGP20}, where the algorithm (and not the adversary!) controls which nodes are asleep and at what time.

\paragraph{Distributed Quantum Computing.}
The seminal work of \cite{DBLP:conf/podc/ElkinKNP14} took the first step towards understanding the time complexity of distributed quantum algorithms for graph problems.
Their results implied that there is little hope for a quantum speed-up for several fundamental problems in the $\congest$ model, such as computing minimum spanning trees, shortest paths, and minimum cuts. 
On the other hand, a quantum advantage is possible for certain graph problems: For instance, \cite{le2018sublinear} demonstrated a quantum speed up for computing the graph diameter and \cite{fraigniaud2024even,van2022framework,censor2022quantum} showed that quantum algorithms can obtain a faster time complexity for subgraph detection.
Apart from the aforementioned works of  \cite{dufoulon2025quantum,gall2025exponential}, we are not aware of other prior results on the message complexity of distributed quantum algorithms for graph problems.
The related (but not equivalent) metric of \emph{communication cost per node} was studied in the context of fault-tolerant consensus by \cite{DBLP:conf/icalp/HajiaghayiKO24}, who give a quantum algorithm that attains a polylogarithmic time complexity while limiting  the number of communication bits sent per node to $n^{\epsilon}$, for any constant $\epsilon>0$. 
Note that an algorithm with low communication cost per node may still incur a high message complexity if each node broadcasts messages across its incident communication links, even if the messages themselves are short. 

\paragraph{Distributed Computing with Advice}
Understanding the impact of advice on the performance of distributed algorithms has seen significant interest over the last two decades. Apart from the already mentioned prior work on the wake-up problem~\cite{fraigniaud2006oracle,wakeup}, the work of \cite{fraigniaud2009distributed} investigated whether advice can break the seminal $\Omega\lt( \log^*n \rt)$-round lower bound of \cite{linial1992locality} for computing a 3-coloring of a cycle. 
More recently, \cite{balliu2025distributed} studied the amount of advice required for solving a graph problem with a time complexity that is a function of the maximum node degree. 
In particular, \cite{balliu2025distributed} show that a single bit of advice per node is sufficient for locally-checkable labeling (LCL) problems~\cite{naor1993can} in graphs exhibiting sub-exponential growth. 
The concept of advice is related to proof labeling schemes~\cite{DBLP:journals/dc/KormanKP10,feuilloley2021redundancy}, which are also known as locally checkable proofs~\cite{goos2011locally}, where a centralized prover assigns labels (``advice'') to the nodes and a distributed verifier must verify a certain graph predicate. 

\subsection{Contributions and Technical Overview} \label{sec:contrib}

We give the first distributed quantum algorithm for the wake-up problem in the port numbering model, obtaining a polynomial improvement in the message complexity compared to the lower bound of \cite{wakeup} known to hold in the classical $\local$ model. We state the time in terms of the \emph{awake distance} $\arad$, introduced in \cite{wakeup}, which is defined as the maximum (over all nodes) shortest path distance of any sleeping node to its closest initially awake node.

\begin{theorem} \label{thm:q_ub_advice}
Consider the synchronous port numbering $\congest$ model, and suppose that an oracle assigns up to $\alpha\le\log_2n$ bits of advice per node.
There is a quantum distributed algorithm that solves wake-up with high probability with a message complexity of  $O\lt( \sqrt{\frac{n^3}{2^{\max\{\lt\lfloor (\alpha-1)/2 \rt\rfloor,0\}}}}\cdot\log n \rt)$ and a time complexity of $O\lt( \arad \cdot n^2\log n \rt)$ rounds.
\end{theorem}

Our advising scheme proceeds in several epochs, each consisting of $n$ phases, and consists of the oracle's computation of the advice followed by a distributed algorithm leveraging this advice to wake up all nodes.
The oracle knows which nodes are awake initially, and each awake node $v$ with ID $i$ will become an actor in phase $i$ of the first epoch.
The goal is to ensure that $v$ has woken up all (currently) sleeping neighbors with high probability before the start of phase $i+1$.
If $v$ has many sleeping neighbors, then the oracle advice simply instructs $v$ to find all sleeping neighbors. 
Otherwise, we need a more careful approach:
The oracle chooses $v$'s advice such that $v$ can narrow down the search space (over its ports) when finding at least one sleeping neighbor $w$ via Grover search. 
Moreover, the oracle deposits additional so called ``proxy advice'' at $w$ that $v$ can use to identify a disjoint search range for finding additional sleeping neighbors. This process continues until all its neighbors are awake.  
The oracle can anticipate the set of sleeping neighbors that are likely to be awoken by each node, and these nodes may become actors in the next epoch.

In Section~\ref{sec:kt0_quantum_lb}, we show that the message complexity of our algorithm is tight up to a logarithmic factor when nodes do not have advice:  

\begin{theorem} \label{thm:lb} 
Any quantum distributed algorithm that solves the wake-up problem (without advice) with sufficiently large constant probability, requires $\Omega\lt( n^{3/2}\rt)$ messages in the port numbering $\local$ model.
This implies the same lower bound for constructing a spanning tree and single-source broadcast, assuming adversarial wake-up. 
\end{theorem}

Arguably, the most natural approach for proving the claimed lower bound would be to construct a lower bound graph where many nodes need to find the right port leading to their (unique) sleeping neighbor, and then leverage an existing lower bound for quantum algorithms: In particular, when considering $\Theta\lt( n \rt)$ independent instances of the unstructured search problem in the sequential quantum query model, where the goal is to find the single marked element for each instance, there are known direct product theorems~\cite{klauck2007quantum} that yield a query complexity of $\Omega\lt( n^{3/2} \rt)$.
An obvious choice would be to associate each search instance with the port mapping of a particular node, whereby the ``hidden'' sleeping neighbor corresponds to the marked element. 
However, simulating the quantum routing mechanism of an algorithm designed for the port numbering model in the query model efficiently (i.e., without performing too many queries per simulated message) requires the simulator to have sufficient knowledge about the port assignments of the nodes.

It is straightforward to implement this simulation for \emph{classical} algorithms and obtain a near one-to-one correspondence between message complexity and query complexity. However, a concrete technical challenge that we face when considering quantum algorithms, is that a node $v$ may send a single message in superposition over all incident ports, and the simulator must ensure that the contents of the send registers are routed correctly into the corresponding receive registers.
Note that the simulator needs to achieve this without exceeding $O(1)$ quantum queries to the input instances, since the single message (sent in superposition) contributes only $1$ to the message complexity.  
It may appear that we can circumvent this obstacle by fixing the port assignments for most pairs of adjacent nodes a priori, and only rewiring one port per node by directing it to a sleeping node.
Unfortunately, this approach would break the assumed independence of the port assignments of distinct nodes, and we are not aware of direct product results for quantum search that yield sufficiently strong lower bounds under such dependencies. 

To avoid the above pitfalls, we first define the $\desc$ problem in the standard quantum query model, which, intuitively speaking, requires the algorithm to recover the parity bits of each element of a permutation that is an involution and does not have any fixed points. 
By leveraging the known hardness of recovering the parity bits of an (arbitrary) permutation due to \cite{van2020quantum}, we obtain the same lower bound for the query complexity of $\desc$.
Returning to the distributed setting, we introduce a class of lower bound graphs $\mathcal{H}_n$, where each $G \in \mathcal{H}_n$ is uniquely identified by a certain hidden perfect matching on a clique subgraph, which is unknown to the nodes a priori.
Moreover, finding this matching turns out to be necessary for achieving wake-up. 
This motivates defining a helper problem called $\matching$, which requires at least half of the clique nodes to find their partner in the matching.
The main technical challenge in the proof is to show how we can simulate a given distributed algorithm for $\matching$ in the query model.
In particular, this allows us to solve the $\desc$ problem with a query complexity that is upper-bounded by the message complexity of the assumed distributed algorithm.

\section{Preliminaries} \label{sec:prelim}
\subsection{Distributed Quantum Computing Model} \label{sec:qmodel}

We consider a distributed synchronous network of $n$ nodes with an arbitrary connected communication topology, modeled as an undirected connected graph $G$.
Each node $v$ has a unique identifier (ID) from $[n]$ and, apart from its ID, only knows $n$ and its degree $\deg(v)$ in the network. Nodes communicate with their neighbors via ports.
The assumption that a node does not know the actual endpoints of its ports is commonly called \emph{port numbering model}~\cite{peleg,DBLP:journals/csur/Suomela13} or $\kt_0$~\cite{AGPV88}.

At the start of each round, each awake node $v$ performs some local computation based on its current state and a source of randomness,\footnote{Our algorithms work as long as each node has a private source of random bits, whereas our lower bound holds even when nodes have access to shared randomness.} which, in turn, determines the communication that it may want to send to its neighbors in $G$ in this round.
In the $\congest$ model, all messages are limited to $O(\log n)$ bits (or qubits), while there is no such restriction in the $\local$ model; we refer the reader to \cite{peleg} for more details and justifications for these assumptions.
Even though we consider a synchronous model, we do not assume that nodes have access to a global clock. 

In more detail, our communication is based on the quantum routing communication model of \cite{dufoulon2025quantum,gall2025exponential}.
We now explain how we extend the quantum routing model to adequately capture the behavior of sleeping nodes in a distributed quantum network: 
We assume that every pair of neighbors in $G$ has a standard communication link for sending (classical) bits and, separately, also has a quantum communication channel for sending messages consisting of qubits.
Conceptually, this means that each node $v$ is equipped with two separate hardware \emph{network interface cards} (NICs): $v$'s classical NIC controls $v$'s classical channels and its quantum NIC determines $v$'s quantum communication.
For simplicity, we assume that $v$'s port mapping used by the quantum NIC and the classical NIC are identical.

A node may be either \emph{asleep} or \emph{awake}. 
An adversary determines the subset of initially awake nodes, who start executing the algorithm from round $1$ onward. 
As explained above, all remaining nodes awaken only upon receiving a message over a {classical} communication link from some already awake neighbor. 
That is, if a sleeping node $u$ is the destination of a classical round-$r$ message $m$, then $u$ wakes up at the start of round $r+1$ and has access to the content of $m$.
A crucial distinction between the node itself and its NICs is that the latter are assumed to be ``on'' and in some initial state even if the node itself is still asleep. 
Concretely, a node $v$ has $4\deg(v)$ \emph{communication registers}, consisting of  the quantum emission register $\EReg_{v\to i}$, the quantum receive register $\RReg_{v\gets i}$, as well as the classical emission register $\ECReg_{v\to i}$ and classical receive register $\RCReg_{v \gets i}$, for each $i \in [\deg(v)]$.
To perform the actual communication step during a round, the quantum routing model swaps the contents of registers via standard controlled swap operations. Note that some of these registers may hold the empty message (i.e., $\perp$).

A sleeping node $u$'s quantum NIC is configured to reflect a phase of $-1$ when receiving a quantum message and does not interact with $u$ itself, i.e., $u$ continues to be asleep. 
Awake nodes, on the other hand, reflect a phase of $1$ (i.e., no phase flip) upon receiving a quantum query message. In addition to responding to query messages, awake nodes are free to send any arbitrary message of $O\lt( \log n \rt)$ qubits over each quantum channel, as instructed by the node.
We again emphasize that, while both NICs of a node are always available for receiving messages, a crucial assumption is that a sleeping node only wakes up upon receiving a message over one of its classical communication links.
Any message sent over a quantum channel to a sleeping node $u$ does \emph{not} awaken the recipient and, as mentioned above, simply causes the quantum NIC of $u$ to trigger a $-1$ phase reflection in response.

\paragraph{Message Complexity of Quantum Algorithms.}
We define the message complexity analogously to  \cite{dufoulon2025quantum,gall2025exponential} as follows:
Consider a synchronous distributed quantum algorithm $\mathcal A$ and a set of initially awake nodes $A_*$. At the beginning of any round $r$, the global network state is a joint quantum state over all local registers and communication registers, and may be in superposition over different local configurations and communication patterns. 
We say that $\mathcal A$ has \emph{quantum message complexity at most $\mu_r(A_*)$ in round $r$} if, for every 
possible configuration, the total number of quantum messages transmitted/received during round $r$ is at most $\mu_r$. 
Moreover, we say that the \emph{quantum message complexity of $\mathcal A$}  is at most $\mu$ with probability $p$, if, for any choice of initially awake nodes $A_*$, it holds that $\sum_{r=1}^R \mu_r(A_*) \le \mu$ with probability at least $p$, where $R$ denotes the worst case number of synchronous rounds until $\mathcal{A}$  terminates. %
Finally, we define the \emph{message complexity of algorithm $\mathcal{A}$} as the sum of its classical and quantum message complexity.

\paragraph{Oracle and Advice.}
We assume that an \emph{oracle} observes the entire network topology including all port mappings and the set of initially awake nodes, and then assigns a bit string, called \emph{advice}, of at most $\alpha$ (classical) bits to each node with the goal of reducing the number of sent messages. 
It is not too difficult to solve wake-up with $\alpha\ge\log_2n$ bits of advice per node and optimal (i.e., $O(n)$) message complexity. 
In fact, this holds even without quantum capabilities and even without assuming that the oracle knows the initially awake nodes; see \cite{wakeup,fraigniaud2006oracle}.
Therefore, we are mostly interested in the case where $\alpha \in o(\log n)$, which means that the oracle cannot even assign as little information as a single integer from $[n]$ (e.g., a port number) to a node.

\paragraph{Quantization of Classical Distributed Algorithms.}
Similarly to \cite{dufoulon2025quantum}, we describe parts of our algorithms as classical algorithms. 
This is valid due to their Lemma~3.1, restated below for completeness:

\begin{lemma}[\cite{dufoulon2025quantum}] \label{lem:quantization}
Let $A$ be a randomized or quantum distributed algorithm, possibly with intermediate measurements. Then, there is a quantum distributed procedure $B$, without any intermediate measurement that simulates $A$ with the same round and message complexities.
\end{lemma}

Our algorithm makes use of the following quantum search subroutine from previous work:

\begin{lemma}[see Theorem~4.1 in \cite{dufoulon2025quantum}] \label{lem:grover}
Consider a node $u$, any subset $X \subseteq \set{1,\dots,\deg(u)}$, and let $f : X \to \set{0,1}$ be a function such that $u$ can evaluate $f(x)$ in $O(1)$ rounds by sending $O(1)$ messages, for any given input $x \in X$. 
 Let $C_u = |f^{-1}(1)|$.
 There exists a distributed quantum algorithm $\qsearch(f,X)$ invoked by $u$ such that, with high probability, if $C_u > 0$,
 then $\qsearch(f, X)$ returns an %
 element $x$ from $f^{-1}(1)$, otherwise it returns $\noneawake$. This holds as long as $f$ does not change throughout the invocation. The quantum message complexity and time complexity  are 
$O\lt(\sqrt{\frac{|X|}{\max\{1, C_u\}}}\cdot\log n\rt)$. %

\end{lemma}

\newcommand{\Bcast}{B_0}
\section{The Upper Bound} \label{sec:q_adv}
Our wake-up algorithm consists of two parts:
In Section~\ref{sec:q_adv_scheme_desc}, we describe how the oracle computes the individual advice assigned to the nodes.
The advice strings are used by the distributed algorithm in Section~\ref{sec:q_adv_algo}.

\subsection{Computing the Advice} \label{sec:q_adv_scheme_desc}
Our advising scheme assumes that the oracle is given the set of initially awake nodes, and that each node obtains at most $\alpha\ge 0$ (classical) bits of advice per node. Throughout this subsection, we assume that $\alpha \ge 3$. When describing the algorithm in Section~\ref{sec:q_adv_algo}, we also explain how to handle the special case that $\alpha\le2$. Note that the oracle simply does not assign any advice in the latter case.
 
The advice string will consist of two equal-sized parts and a single bit, which motivates defining parameter 
\begin{align}
\beta = \max\Set{\lt\lfloor \frac{\alpha-1}{2} \rt\rfloor,0}, \label{eq:q_beta} 
\end{align}

Recall that the node IDs are a permutation of $[n]$. 
The algorithm that uses the advice performs $\arad$ epochs and, in each epoch, processes the nodes in increasing order of their IDs in a sequential greedy-like fashion. 
That is, there are $n$ phases in an epoch, where in the $j$-th phase, only the node $v$ with ID $j$ can initiate some computation, as described below, and, in that case, we call $v$ the \emph{actor of phase $j$}. 

We now describe how the oracle determines the advice string for each node.
Note that the distributed quantum algorithm using this advice (see Sec.~\ref{sec:q_adv_algo}) makes use of Lemma~\ref{lem:grover_iter} and hence necessarily is randomized; consequently, the oracle cannot deterministically predict the nodes’ states throughout the algorithm’s execution. 

Let $N_v$ denote the set of all neighbors of $v$. 
The oracle defines several sets inductively over the epochs $i=1, \ldots, \arad$ as follows: 
\begin{itemize} 
\item 
$A_0$ is the empty set and 
$A_1$ is the set of initially awake nodes. 
\item $S_i = \bigcup_{v \in A_i} N_v \cap \lt(V(G)\setminus \lt(A_i\cup A_{i-1}\rt)\rt)$.
\item For $i\ge1$, we define $A_{i+1} = S_i$. 
\item 
For $i \ge 1$, let $v_1^{(i)}, \ldots, v_{|A_i|}^{(i)}$ be the elements in $A_i$, ordered in increasing order of ID. We define
\begin{align}
    S_{v_j}^{(i)} &= N_{v_j^{(i)}} \cap \lt( S_i  \setminus \bigcup_{k \in [j-1]}S_{v_k}^{(i)} \rt). \label{eq:qsi}  
\end{align}
\end{itemize}
Intuitively, $A_i$ is the set of newly awake nodes (that is, the nodes in $A_i$ that were asleep at epoch $i-1$). We call $A_i$ the \emph{actors of epoch $i$}. 
Our advising scheme assigns advice to the actors of epoch $i$ so that, collectively, they can wake up all their neighbors that are still asleep at the start of epoch $i$. Assuming that the algorithm succeeded so far, the set $S_{v_j}^{(i)}$ contains exactly the set of sleeping neighbors of $v_j$ when it becomes actor. %
In particular, this implies that a node in $A_i$ can only have neighbors in $A_{i-1} \cup A_i \cup A_{i+1}$.
We state some immediate properties of these sets:

\begin{lemma}\label{lem:epoch}
Consider any node $v$. The following hold:
\begin{enumerate}
    \item $v$ is an actor of at most one epoch. That is, there is a unique $i$ such that $v$ is in $A_i$.
    \item $S_v^{(i)} \cap S_{w}^{(j)} = \emptyset$ for any nodes $w \neq v$, where $i$ respectively $j$ are their (unique) epochs in which they are actor.
\end{enumerate}
\end{lemma}

The advice of $v$ consists of a triple $(g_v,\Lambda_v,\Pi_v)$, where $g_v$ is just a single bit, while $\Lambda_v$ and $\Pi_v$ have a length of at most $\beta$ bits each.
According to the definition of $A_i$ above, each node $v$ is an actor in at most one epoch $i$. 
The bit string $\Lambda_v$ will enable $v$ to narrow down the search range for finding at least one neighbor that is asleep at the start of epoch $i$. 
On the other hand, $\Pi_v$ consists of advice that will allow some already-awake neighbor $w$ of $v$ to find another sleeping node, which motivates calling $\Pi_v$ the \emph{proxy advice carried by $v$}.

\paragraph{Advice Tree and Associated Port Ranges:}
For each $v$, we conceptually define an advice tree $\mathcal{T}_v$ that is used by the oracle for computing $\Lambda_v$ and $\Pi_v$. 
Let $n_v$ be the degree of $v$ in the network, and let $n_v'$ be the smallest power of $2$ such that $n_v \le n_v'$.
We make use of a directed complete binary tree $\mathcal{T}_v$ of depth exactly $\log_2n_v'$, called \emph{advice tree}, whose edges are directed from parents towards their children, i.e., there is exactly one directed path from the root to each leaf.
We attempt to label the leaves of $\mathcal{T}_v$ by the integers in $[n_v]$ in strictly increasing order. Such a labeling is possible if $n_v=n_v'$; in the case that $n_v<n_v'$, we label all of the remaining $n_v'-n_v$ right-most leaves with $n_v$.
We leverage the natural one-to-one correspondence between bit strings of length at most $\lt\lceil \log_2 n_v \rt\rceil \le \log_2 n_v'$ and directed paths starting from the root:
That is, the $k$-th bit of the given bit string determines whether we follow the left or right branch when reaching a level-$k$ vertex of $\mathcal{T}_v$, with the root being on level $1$.
Thus, for a given $\ell$-bit string $b$, we define the \emph{port range associated with $b$} as the interval $[p_1,p_2] \subseteq [n_v]$, such that $p_1$ and $p_2$ are the smallest, respectively largest, label of any leaf reachable from the endpoint of the length-$\ell$ directed path that corresponds to $b$. 
Clearly, for any neighbor $w$ of $v$ and any $k \le \log_2 n_v'$, there exists a $k$-length directed path from the root of $\mathcal{T}_v$ whose corresponding bit string has an associated port range containing $v$'s port to $w$, denoted by $p_{v,w}$.
We call the associated port range the \emph{level-$\beta$ port range of $w$ with respect to $v$} and denote it as $P_v^\beta(w)$. 
Note that if $p_{v,w}=n_v$ and $n_v<n_v'$, there may be more than one such port range (since multiple leaves are labeled with $n_v$), and, in that case, we simply pick the largest port range to avoid ambiguity.

\paragraph{The Oracle's Algorithm:}
The oracle uses $\mathcal{T}_v$ to compute the advice associated with $v$ as follows:
Let $i$ be the epoch where $v$ is an actor, i.e., $v \in A_i$. 
\begin{itemize}
\item We set the bit $g_v=1$ if and only if $|S_v^{(i)}| \ge 2^{\beta}$.  
\item 
If $g_v=1$, then $\Lambda_v$ is the empty string. 
Moreover, $\Lambda_v$ is also empty if $g_v=0$ and $S_v^{(i)}=\emptyset$. Otherwise, the oracle determines $\Lambda_v$ such that it is associated with the level-$\beta$ port range $P_v^\beta(w_1)$, where $w_1$ denotes the node with the smallest ID in $S_v^{(i)}$.

\item 
Next, the oracle inductively defines a sequence of nodes $w_1, \ldots, w_k \in S_v^{(i)}$ such that each of these nodes will carry proxy advice that enables $v$ to find all nodes in $S_v^{(i)}$.  
   \begin{itemize} 
   \item We have already defined $w_1$ and $P_v^\beta(w_1)$ in the previous bullet point; we also define the set $S_1'=S_v^{(i)}$.

    \item For $j\ge 2$, we obtain $S_j'$ by removing all nodes from $S_{j-1}'$ that $v$ is connected to via a port in $P_v^\beta(w_{j-1})$, which includes $w_{j-1}$ itself at the very least.
    Let $w_j$ be the node with the smallest ID in $S_j'$. 
    We determine the proxy advice $\Pi_{w_{j-1}}$ assigned to $w_{j-1}$ by setting $\Pi_{w_{j-1}}$ such that the associated level-$\beta$ port range is $P_v^\beta(w_{j})$; we simply leave $\Pi_{w_j}$ empty if $S_j'=\emptyset$.
    \end{itemize}
\end{itemize}
This concludes assigning the advice with respect to $v$.

\subsection{Toolbox: Iterated Quantum Search} \label{sec:iterated_grover}

In our distributed algorithm described in Section~\ref{sec:q_adv_algo} below, we frequently call the $\qsearch$ procedure (see Lemma~\ref{lem:grover} in Section~\ref{sec:prelim}) to exhaustively find all marked elements within a certain search range. 
The following lemma summarizes the complexity bounds for these iterative calls:

\begin{lemma}\label{lem:grover_iter}
Based on the procedure $\qsearch$ stated in Lemma~\ref{lem:grover}, 
there exists a distributed algorithm called $\iterqsearch$, invoked by $u$, that
returns \textit{all} elements $x \in X$ with $f(x)=1$ with high probability.
Its round and (quantum) message complexity are
$O(\sqrt{|X|  \max\set{C_u,1}}\cdot \log n)$. %
\end{lemma}
\begin{proof}
The algorithm proceeds in iterations.
In each iteration, $u$ invokes $\qsearch(f,X)$ as given by Lemma~\ref{lem:grover}.
In the case $C_u=0$, this returns $\emptyset$ (w.h.p.) and takes $O(\sqrt{|X|} \cdot  \log n)$.

Thus, we focus on the case $C_u>0$:
If $\qsearch$ returns some $x\in X$ with $f(x)=1$ in some iteration, then $u$ locally removes $x$ from $X$ and the 
next iteration starts. 
If $\qsearch$ returns \noneawake, we simply terminate.

After the $i$-th successful iteration, there are $C_u - i + 1$ remaining 
marked elements. By Lemma~\ref{lem:grover}, the round and message 
complexity of that call to $\qsearch$ is
$O (\sqrt{\frac{|X|}{C_u - i + 1}}\log n).$
Summing over all $C_u$ successful calls, the total round- as well as message 
complexity is $$O\lt(\log n\sum_{i=1}^{C_u}\sqrt{\frac{|X|}{C_u-i+1}}\rt) =O\lt(\log n \sqrt{|X|C_u}\rt).$$
The final call that returns $\noneawake$ adds at most 
$O\lt(\sqrt{|X|}\log n\rt)$, which does not change the asymptotic bound.

We choose the failure probability of $\qsearch$ to be at most
$1/n^{c+2}$ for some constant $c>0$. 
As there are at most $|X|+1\le n+1$ calls in total, a union bound ensures that the overall failure probability is at most $1/n^c$. 
\end{proof}

\subsection{The Distributed Wake-up Algorithm} \label{sec:q_adv_algo}
We now describe how the nodes use the provided advice to solve the wake-up problem.

In each epoch $i$, there are $n$ phases, and only the node with ID $l$ can be an actor in phase $l$ of some epoch. 
Since each initially awake node knows its own ID and can keep track of the  number of passed rounds since it woke up in a local clock, it also knows in which phase it can become the actor.
Moreover, by sending along the node's local clock value in every wake-up message that it may send, we can ensure that every awake node has the same notion of time. 
Note that if a node $v$ with ID $l$ wakes up only after the start of an epoch $i$ (i.e., due to receiving a message), then it assumes to have missed its turn and does not become an actor in the current epoch, even if phase $l$ has not yet started; however, it will become the actor in phase $l$ of the next epoch. 
That is, node $v$ has become an actor in phase $l$ of epoch $i+1$, and it no longer becomes actor in any subsequent epoch.
All initially awake nodes are actors in epoch $1$. 

Consider the start of phase $l$ and suppose that node $v$ is an actor and has ID $l$. We proceed as follows: %
\begin{itemize} 
\item If $g_v=1$ or the node did not receive any advice at all (i.e., $\alpha \le 2$), then $v$ simply performs $\iterqsearch$ on its entire port range (i.e., $[\deg(v)]$) to wake up all neighbors that are asleep (see Lemma~\ref{lem:grover_iter}). 

\item If $g_v=0$ and $\Lambda_v$ is empty, then $v$ skips its actor phase. Note that this captures the case when $v$ does not have any sleeping neighbors at the point when it becomes actor.

\item On the other hand, if $g_v=0$, then $v$ uses the advice in $\Lambda_v$ as well as the proxy advice deposited at its neighbors to narrow down the search range for its asleep neighbors:
In more detail, $v$ first runs $\iterqsearch$ (see Lemma~\ref{lem:grover_iter}) on the port range $[a,b]$ associated with $\Lambda_v$ to wake up every sleeping node $s_1,\dots,s_\ell$ that $v$ is connected to via a port in $[a,b]$.
According to the advising scheme (see Section~\ref{sec:q_adv_scheme_desc}), one of these sleeping neighbors holds proxy information that it sends to $v$ in response to receiving the (classical) wake-up message.
Once $v$ has woken up everyone in $[a,b]$, it uses the received proxy information to identify the next port range $[a',b'] \subseteq [n_v]$ in the advice tree $\mathcal{T}_v$, which is guaranteed to be disjoint from $[a,b]$ (see Lemma~\ref{lem:p_range}). 
Then, $v$ again calls $\iterqsearch$ to find every sleeping neighbor, in sequence, which may again reveal a different port range, and so forth. 
We show that this wakes up all neighbors of $v$.
As mentioned above, $v$ does not participate as actor in any future epoch.
\end{itemize}
To ensure synchronicity of the nodes' local clocks, each phase lasts exactly $\tau = C {n} \log n$ rounds for a sufficiently large constant $C$ that depends on the asymptotic running time of $\iterqsearch$. 
 
\subsection{Analysis} \label{sec:q_adv_analysis}

\begin{lemma} \label{lem:p_range}
Consider any node $v$.
Every level-$\beta$ port range associated with $\Lambda_v$ or $\Pi_v$ has a size of at most $\frac{n}{2^{\beta-1}}$, and all of these ranges are pairwise disjoint.
\end{lemma}
\begin{proof}
Disjointness follows immediately from the description of the advice tree.
Since every nonempty string $\Lambda_v$ and $\Pi_v$ has a length of $\beta$ bits, the corresponding path in $\mathcal{T}_v$ reaches a vertex $x$ on level $\min\set{\beta,\log_2\deg(v)}$.
Recall from Section~\ref{sec:q_adv_scheme_desc} that $\mathcal{T}_v$ has exactly $n_v'$ leaves, where $n_v'$ is the smallest power of $2$ that is at least as large as $\deg(v)$.
By an easy counting argument, it follows that the number of leaves with distinct labels that are reachable from $x$ is at most $\frac{n_v'}{2^{\min\set{\beta,\log_2\deg(v)}}} \le \frac{2n}{2^{\beta}}$.
\end{proof}

\begin{lemma} \label{lem:q_adv_one_phase}
Suppose that node $v$ is an actor of epoch $i$. 
Then, with high probability,  $S_v^{(i)}$ is the set of sleeping neighbors of $v$ at the beginning of phase $\id(v)$ of epoch $i$ and all neighbors of $v$ are woken up by the end of phase $\id(v)$ of epoch $i$. 
\end{lemma}
\begin{proof}
In the case that $v$ was given the advice $g_v=1$, the claim is immediate from the correctness of $\iterqsearch$, which $v$ iteratively invokes until it has found all sleeping neighbors. 

Now assume that $g_v=0$.  According to the description of the algorithm, $v$ uses the advice string $\Lambda_v$ to determine a port range $[p_1,p_2]$, and then invokes $\iterqsearch([p_1,p_2],f)$ to find all sleeping neighbors $w$.
As a result of awakening $w$, node $v$ may receive proxy advice $\Pi_w$ from at least one neighbor $w$ (unless all of its sleeping neighbors are already in $[p_1,p_2]$), which, by the description of the advising scheme, refers to another port range $[p_1',p_2']$ that, according to Lemma~\ref{lem:p_range}, is disjoint from $[p_1,p_2]$.
Again, $v$ invokes $\iterqsearch$ to wake up at least one sleeping neighbor in $[p_1',p_2']$ and so forth.
\end{proof}

By Lemma~\ref{lem:epoch}, each node is an actor in exactly one phase of exactly one epoch.
Hence, we can determine the message complexity bound by analyzing the messages sent during one particular phase.
In the following, we separately analyze the message complexity for the nodes with advice $g_v=0$ and the ones having $g_v=1$. 

\begin{lemma} \label{lem:q_adv_msgs_g0}
Let $B_0$ be the set of nodes such that for all $v \in B_0$, we have $g_v=0$. Then the message complexity during all phases in which such a $v$ is actor is $O\lt( \sqrt{\frac{n^3}{2^\beta}}\log n\rt)$ with high probability. 
\end{lemma}
\begin{proof}
Suppose that $v$ is an actor of some epoch $i$ and is given advice $g_v=0$.
Let $k = |S_v^{(i)}|$. 
We will show that the total number of messages sent by $v$ throughout the execution is at most $O\lt( k \sqrt{\frac{n}{2^\beta}}\log n\rt)$, with high probability. 
Let $r$ be the starting round of the phase in which $v$ is the actor. 
By the description of the algorithm, starting from round $r$, node $v$ invokes $\iterqsearch([p_1,p_2],f)$ where $[p_1, p_2]$ is the port range derived from $\Lambda_v$. Subsequently, it invokes $\iterqsearch([p_1', p_2'])$ where $[p_1', p_2']$ is the port range derived from the proxy advice $\Pi_w$ received from a (previously sleeping) neighbor $w \in S_v^{(i)}$, whose port was in $[p_1', p_2']$. Node $v$ repeats this process until it receives no new proxy advice from its neighbors. 
Let $\ell$ be the total number of times $v$ invokes $\iterqsearch$. It follows from Lemma~\ref{lem:q_adv_one_phase} that, with high probability, $\ell \le k$. %
Let $\iterqsearch_j$ denote the $j$-th invocation of $\iterqsearch$ by $v$, and $k_j$ be the number of sleeping nodes in the port range used by $\iterqsearch_j$. We have $\sum_{j=1}^{\ell} k_j = k.$
Since each invocation of $\iterqsearch$ operates on a search range determined by the associated port range of the respective proxy advice, it follows from Lemma~\ref{lem:p_range} that the size of each range is at most $2n/{2^{\beta}}$. 
By Lemma~\ref{lem:grover_iter}, the message complexity of $\iterqsearch_j$ is $O(\sqrt{\frac{n}{2^\beta}k_j}\log n)$ with high probability. 
Summing over all $\ell$ invocations and taking a union bound, the total message complexity is (w.h.p.)
\begin{align}
O\lt( \sqrt{\frac{n}{2^\beta}}\sum_{j=1}^{l}\sqrt{k_j}\log n\rt)  
& =  
O\lt(\sqrt{\frac{n}{2^\beta}}\sqrt{lk}\log n\rt)   \label{eq:q_adv_cs} \\
\ann{since $l \le k+1$}
&=
O\lt(\sqrt{\frac{n}{2^\beta}}k\log n\rt),\notag
\end{align}
where we have used the Cauchy-Schwarz inequality in \eqref{eq:q_adv_cs}. 

To complete the proof, we note that, by Lemma~\ref{lem:epoch}, each $v \in B_0$ wakes up a distinct set of sleeping neighbors.
Given that there are at most $n$ sleeping neighbors in total, the total message complexity incurred by all nodes in $B_0$ during the phase where they are active is $O\lt( \sqrt{\frac{n^3}{2^\beta}}\log n\rt)$ (w.h.p.) as claimed.
\end{proof}

\begin{lemma} \label{lem:q_adv_msgs_g1}
Let $B_1$ be the set of nodes such that, for every $v \in B_1$, it holds that $g_v=1$.
Then, the message complexity over all phases in which one of the nodes in $B_1$ is actor is at most %
$O\lt( \sqrt{\frac{n^3}{2^{\beta}}}\log n \rt)$ with high probability.
\end{lemma}
\begin{proof}
Recall that each node $v$ is associated with a subset $S_v^{(i)}$ for an epoch $i$. 
Given that $g_v = 1$, we have $|S_v^{(i)}| \geq 2^\beta$. 
By Lemma~\ref{lem:epoch}, since all subsets $S_v$ for $v \in B_1$ are disjoint, we have 
\begin{align}
|B_1| \le \frac{n}{2^\beta}. \label{eq:q_adv_B}
\end{align}
Recall that $v$ invokes $\iterqsearch$ (on its entire port range). 
By Lemma~\ref{lem:grover_iter}, all asleep neighbors of 
$v$ are awakened with high probability, and this incurs a message complexity of $O\lt(\sqrt{nk_v}\cdot \log n\rt).$
Summing up over all phases in which a node in $B_1$ is actor and taking a union bound, we have the total message complexity is (w.h.p.)
\begin{align}
    O\lt(\sum_{v \in B_1} \sqrt{k_vn}\cdot\log n \rt) 
    &= 
   O\lt( \sqrt{n} \sum_{v\in B_1} \sqrt{k_v}\cdot \log n\rt) \notag \\ 
   \ann{by Cauchy-Schwarz ineq.}
   &=
  O\lt(\sqrt{n|B_1|\sum_{v \in B_1}k_v }\cdot \log n \rt) \notag \\ 
  \ann{by \eqref{eq:q_adv_B}}
  &= 
  O\lt(\sqrt{\frac{n^3}{2^\beta}}\log n \rt).
\end{align}
\end{proof}

\begin{lemma} \label{lem:q_adv_time}
The algorithm terminates in $O\lt(\arad \cdot n^2\log n\rt)$ rounds.
\end{lemma}
\begin{proof}
This follows directly by recalling that each phase takes $O\lt( n\log n\rt)$ rounds, there are $n$ phases per epoch, and all nodes are awake after $\arad$ epochs with high probability.
\end{proof}

\newcommand{\req}{\mathsf{req}}
\newcommand{\search}{\mathsf{Search}}
\section{A Lower Bound for Distributed Quantum Algorithms} \label{sec:kt0_quantum_lb}

Throughout this section, we assume that $n$ is a multiple of four and that the network has $2n$ nodes. 
The first step in obtaining the lower bound for the wake-up problem, is a reduction between the wake-up problem and finding a ``hidden'' perfect matching in a certain class of graphs, which we define as follows: 
The nodes are equally partitioned into the \emph{center nodes} $V=\set{v_1,\dots,v_n}$ and the initially asleep nodes in $W=\set{w_1,\dots,w_n}$.
We fix the IDs of the nodes in $V$ such that $v_i$ has ID $i$, and, similarly, fix the IDs of the nodes in $W$ to an arbitrary permutation of $[2n] \setminus [n]$.
Let $C$ be the graph consisting of the $n$-clique on the nodes in $V$ that we obtain after fixing an arbitrary port assignment $\pi_C$. 
Consider the set $\mathcal{M}$ of all possible perfect matchings on $C$, each of which consists of exactly $n/2$ edges.
Suppose that $M \in \mathcal{M}$ consists of edges $\set{v_{i_1},v_{j_1}},\dots,\set{v_{i_{n/2}},v_{j_{n/2}}}$, whereby all indices $i_k,j_{k'}$ ($k,k' \in [n/2]$) are distinct.
Since the port assignment is fixed, $M$ induces a \emph{port set} $P_M$ of port-pairs associated with the matched edges: $P_M=(p_{i_1},p_{j_1}),\dots,(p_{i_{n/2}},p_{j_{n/2}})$.
That is, $p_{i_k}$ is the port of $v_{i_k}$ connecting to $v_{j_k}$, and $p_{j_k}$ is defined analogously.
We use $\mathcal{H}_n$ to denote the set of all possible lower bound graphs obtained in this manner. 
Given a matching $M \in \mathcal{M}$ and the associated port set $P_M$, we obtain a (unique) lower bound graph $G \in \mathcal{H}_n$ on the nodes $V \cup W$, by modifying the edges of the clique $C$. 
That is, for each $k \in [n/2]$, we do the following:
\begin{itemize} 
\item Remove edge $\set{v_{i_k},v_{j_k}}$.
\item Add edge $\set{v_{i_k},w_{i_k}}$ and let $p_{i_k}$ be $v_{i_k}$'s port connecting to $w_{i_k}$.
\item Add edge $\set{v_{j_k},w_{j_k}}$ and let $p_{j_k}$ be $v_{j_k}$'s port connecting to $w_{j_k}$.
\end{itemize}
As a result of this process, each $w_{i'}$ has exactly one edge and hence it connects to $v_{i'}$ via port $1$.
Since we only consider executions of the assumed distributed algorithm on graphs in $\mathcal{H}_n$, we can assume that every node knows the port assignment function $\pi_C$.
It is straightforward to see that providing additional knowledge to the nodes can only strengthen the lower bound. 
Note that, however, the nodes do not know which perfect matching $M$ was used to construct the specific lower bound graph $G \in \mathcal{H}_n$. We call $M$ a \emph{hidden perfect matching in $G$}.

\begin{fact}
Every node $v_i \in V$ has the same initial state in $C$ as it does in each $G \in \mathcal{H}_n$.
\end{fact}

\paragraph{Finding the Hidden Matching:}
We say that an algorithm solves the $\matching$ problem on a given graph $G \in \mathcal{H}_n$ with hidden matching $M$, if, for every edge $e \in M$, at least one of the two IDs of the endpoints of $e$ is output by its corresponding node in $C$.
For each $v_i \in C$, we use $X_i$ to denote the ID of its matched neighbor in $M$. That is, $\{v_i, v_{X_i}\}$ is an edge in $M$.  
It is straightforward to see that any wake-up (quantum or classical) algorithm must also solve the $\matching$ problem in our lower bound graphs without asymptotic overhead:
Once all nodes are awake, each $w_i$ sends a single message to its neighbor, allowing $v_i$ to learn the port $j$ that leads to $w_i$. Thus $v_i$ can set $X_i = \pi_C(i,j)$.\footnote{Recall that 
$\pi_c$ is fixed and known to all nodes.}
This contributes $n$ additional messages. Since any distributed wake-up algorithm already uses 
$\Omega(n)$ classical wake-up messages, the asymptotic message complexity is unchanged.
 
\begin{fact}
Any quantum wake-up algorithm that has a message complexity of $\mu$ on the graphs in $\mathcal{H}_n$ where the initially awake nodes are the set $V$, gives rise to a $\matching$ algorithm with message complexity $O(\mu)$ that has the same upper bound on the error probability.
\end{fact}

\paragraph{The $\desc$ Problem:}
Our lower bound relies on a known quantum query lower bound for computing the single-bit descriptor of a permutation, which, intuitively speaking, tells us only the parity bit of each element. 
Formally, give a permutation $\sigma$ on $[n]$, the \emph{single-bit descriptor of $\sigma$} is an $n$-length bit string $z$ such that $z_i = \sigma(i) \bmod 2$.
Note that we can represent $\sigma$ as an $n \times n$ permutation matrix ${P}$ that, in every row and column, has exactly one $1$ and is all-zeroes everywhere else, whereby  ${P}_{i,j}=1$ if and only if $\sigma(i)=j$.
The problem input is represented by having access to a \emph{permutation unitary} $O_P$ that satisfies
\begin{align}
	O_P : \ket{i,j,b} \to \ket{i,j,b \oplus P_{i,j}}. \notag
\end{align}
As we will elaborate in more detail below, in the context of our lower bound graph construction, we will use ${P}_{i,j}=1$ to mean that node $v_i$ is matched with node $v_j$.
Observe that a perfect matching imposes additional structure on ${P}$, namely ${P}_{i,j}={P}_{j,i}$ and $P_{i,i}=0$.
In other words, we need to restrict ourselves to permutations that are involutions without fixed points.
Thus, we define problem $\desc$ as computing the single-bit descriptor of a given fixed-point-free involution by performing quantum queries to $O_P$.

\begin{lemma}[see Lemma~16 in \cite{van2020quantum}] \label{lem:sbit}
Computing the single-bit descriptor of a permutation $\sigma$ on $[n]$ with probability at least $\frac{2}{3}$ requires $\Omega\lt( n^{3/2} \rt)$ quantum queries.
\end{lemma}

\paragraph{High-Level Overview of the Lower Bounds:}
The main idea behind our lower bounds is that any distributed algorithm that solves the $\matching$ problem on graphs in $\mathcal{H}_{n}$ can be simulated in the sequential quantum query model~\cite{hamoudi2025brief} for computing the single-bit descriptor of a permutation on $[n]$.
To this end, we first establish a lower bound for computing such a descriptor for a special class of permutations, i.e., fixed-point-free involutions, in Lemma~\ref{lem:involution}.  
Then, we show in Section~\ref{sec:qsim} how to simulate a distributed quantum algorithm with message complexity $\mu$ by using $O(\mu)$ queries in the sequential query model.

\begin{lemma}
\label{lem:involution}
The query complexity of solving problem $\desc$ with probability at least $\frac{2}{3}$ is $\Omega\lt( n^{3/2} \rt)$.
\end{lemma}
\begin{proof}
We will show a reduction to the lower bound in Lemma~\ref{lem:sbit}.
Consider an algorithm $\mathcal{A}$ that solves $\desc$ with query complexity $q$ and succeeds with probability at least $\frac{2}{3}$. 
Let $P$ be the matrix representation of a permutation $\sigma$ on $[n]$, and note that $\sigma$ is not necessarily an involution and may have fixed points. 
We define a fixed-point-free involution $\sigma'$ on $[2n]$ with a $2n\times 2n$ permutation matrix $P'$, and simulate access to $O_{P'}$ to execute algorithm $\mathcal{A}$.
For every $i \in [n]$, it holds that
\begin{align} 
\sigma'(i)   &= n + \sigma(i); \label{eq:sigma1}
\\ 
\sigma'(n+i) &= \sigma^{-1}(i). \label{eq:sigma2}
\end{align}
Note that $\sigma'(\sigma'(i))=\sigma'(n+\sigma(i))=\sigma^{-1}(\sigma(i))=i$, as required.
Assuming that the $\desc$-algorithm $\mathcal{A}$ terminates correctly, we obtain the single-bit descriptor $z'$, which is a bit string of length $2n$.
Finally, we compute the single-bit descriptor $z$ of $P$ by defining $z_i = z'_i \oplus (n \bmod{2}).$
The correctness of the computed string $z$ and the success probability are immediate from \eqref{eq:sigma1}, \eqref{eq:sigma2}, and the fact that $\mathcal{A}$ succeeds with probability at least $\frac{2}{3}$.

It remains to show that we can define $O_{P'}$ such that each query can be simulated by a single query to $O_P$.
Note that $P'=
\lt[\begin{matrix}
0 & P\\
P^T & 0.
\end{matrix}
\rt]
$
Thus, we define 
\begin{align}
	O_{P'}\ket{i,j,b} =
  \begin{cases}
		\ket{i,j,b} & \text{if $((i\le n) \land (j\le n) \lor ((i>n) \land (j>n))$}   \\
		O_P\ket{j,i-n,b} & \text{if $i> n$ and $j\le n$}   \\
		O_P\ket{i,j-n,b} & \text{if $i\le n$ and $j>n$,}   \\
  \end{cases}
\end{align}
which ensures that we need to make at most one call to $O_P$ for each query.
Finally, by Lemma~\ref{lem:sbit}, we conclude that $\mathcal{A}$ requires at least $\Omega\lt( n^{3/2} \rt)$ quantum queries in the worst case.
\end{proof} 

\newcommand{\vac}{\perp}
\newcommand{\inbox}{\textsc{inbox}}
\newcommand{\pre}{U_\textsc{prep}}
\newcommand{\route}{U_\textsc{deliver}}
\newcommand{\oracle}{O}
\newcommand{\xor}{\oplus}
\newcommand{\target}{\mathsf{tgt}}
\newcommand{\send}{\textsc{send}}
\newcommand{\receive}{\textsc{receive}}
\newcommand{\recv}{\textsc{receive}}
\newcommand{\round}{\textsc{round}}
\newcommand{\psend}{\textsc{psend}}
\newcommand{\Id}{\textsc{I}}
\newcommand{\net}{\mathrm{net}}
\newcommand{\outbox}{\textsc{outbox}}
\newcommand{\Msg}{\mathsf{message}}
\newcommand{\mem}{\textsc{memory}}
\newcommand{\push}{U_\textsc{push}^{(r)}}
\newcommand{\SWAP}{\mathrm{SWAP}}

\subsection{Simulating Quantum Routing in the Query Model} \label{sec:qsim}
We now show how to simulate a message-efficient distributed quantum algorithm in the query model with low query complexity. Combining the following Lemma~\ref{lem:qsim} with Lemma~\ref{lem:involution} completes the proof of Theorem~\ref{thm:lb}.

\begin{lemma} \label{lem:qsim}
Any quantum distributed algorithm that solves the $\matching$ problem with a message complexity of $\mu$ on graphs in $\mathcal{H}_n$, gives rise to a query algorithm for the $\desc$ problem with a query complexity of $O(\mu)$ without increasing the probability of error.
\end{lemma}
\begin{proof} 
To prove the lemma, we start with an instance of the $\desc$ problem in the sequential query model, which is given by a permutation matrix $P$, where the values of $P$ are only accessible via queries to the permutation unitary $O_P$. 
We show that we can simulate a quantum distributed algorithm for solving the $\matching$ problem on a graph $G \in \mathcal{H}_n$, where, for each matched edge, one of its endpoints $v_i$ needs to output $X_i$, the ID of the node that matched with $v_i$.  
The edges of the hidden matching that determines $G$ are uniquely defined by the involution corresponding to the input matrix $P$ which is accessible only via the permutation unitary $O_P$. 
That is, $v_i$ is matched with $v_k$ if $P_{i,k}=1.$
Note that we view the unknown matched IDs
$X = (X_1,\ldots,X_{n})$ as fixed but unknown to the algorithm; they are only
accessible via the permutation unitary $O_p$. 
Recall that each node communicates with its neighbors via ports. 
We use $\pi_C(i,j)$ to denote the clique partner of $v_i$ via port $j$ in $C$, determined by the port assignment $\pi_C$. 
Intuitively, when the algorithm chooses to send a message from $v_i$ via its local port $j$, we can determine whether this port leads to $w_i$ or to the clique partner $v_k$, where $k = \pi_C(i,j)$, by querying $O_p$. 
If $P_{i,k}=1$, then $v_i$ is matched with $v_k$, and hence, the destination of the message sent by $v_i$ via port $j$ would reach $w_i$; otherwise, it will reach $v_k$ instead. 

In the following, we will describe how to simulate sending messages in a single round $r$ of the distributed algorithm in the query model. 
Note that, in our simulation, we route classical messages in the same way as quantum messages; however, the simulator tags classical messages so that it can enforce the rule that sleeping nodes can be awakened only by classical messages. In particular, when a quantum message reaches a sleeping node, the node does not wake up and its quantum NIC instead reflects a phase of $-1$. 
This behavior can be simulated using an $X$-independent controlled phase gate applied by the simulator when the destination is an asleep node and the channel is quantum. 
  
In a given round, multiple nodes may send messages on some of their local ports (possibly in superposition); our simulation will handle all such sends. %
To simplify the proof, we will mostly restrict our attention to routing the messages sent by center nodes $V$, whose endpoint depends on the unknown $X$. 
Routing messages sent by the nodes in $W$ can be simulated trivially without queries.

A unitary transformation $U$ on the algorithm's registers is called \emph{$X$-independent} if it is fixed in advance and does not change with the values of $X = (X_1,\ldots,X_{n})$, i.e., the matrix representation of $U$ is the same for all choices of $X$ and hence $U$ does not involve any queries to $O_P$.  

For each node $u \in V \cup W$ with degree $\deg(u)$,  we index its 
ports by $j \in [\deg(u)]$. We distinguish between \emph{port-level}
registers, which the distributed algorithm conceptually uses when it sends on
local ports, and \emph{edge-level} registers, which are used for physical
routing along edges.

\paragraph{Port-send registers.}
For each node $u$ and port $j \in [\deg(u)]$ we have a 
\emph{port-send register} $\psend_{u,j}$, %
with a
distinguished vacuum state $\vac$ orthogonal to all valid message states. 
\paragraph{Edge-level registers.}
We introduce \emph{edge-level registers} $\send_{u\to v}$ and
$\recv_{v\leftarrow u}$ for every ordered pair $(u,v)$ of nodes. For pairs with
$\{u,v\} \in E(G)$, $\send_{u\to v}$ and $\recv_{v\leftarrow u}$ are used by the routing unitary $\route$ to model the delivery of messages along the edge
$(u,v)$. 
On the other hand, if $\{u,v\} \notin E(G)$, then the corresponding registers are never updated and remain in the vacuum state for the entire computation. 
We point out that the simulation algorithm does not know whether $\set{u,v} \in E(G)$ in advance, which is the reason for creating a register for each possible edge.
\paragraph{Outbox registers.} 
We additionally introduce $\mu_r$ \emph{outbox} registers $
\outbox_1,\ldots,\outbox_{\mu_r}$, 
where $\mu_r$ denotes the upper bound on the messages sent by $V$ in round $r$%
, i.e., in every computational-basis configuration, at most $\mu_r$ port-send registers associated to $V$ contain a
non-vacuum message. 
Each $\outbox_t$ is a tensor product
$\outbox_t = I_t \otimes J_t \otimes \Msg_t,$
where $I_t$ is a $\lceil\log n\rceil$-qubit register encoding an index $i\in[n]$,
$J_t$ is a $\lceil\log \deg(v_i)\rceil$-qubit register encoding a port index $j$ of $v_i$,
and $\Msg_t$ is the register to store the message that $v_i$ intends to send on port $j$ in round $r$. 
We use the convention that $\outbox_t=\perp$ denotes an empty entry.

\paragraph{Joint network register.}
Lastly, we denote $\mem_u$ the local memory register of node $u$. 
The joint network register at round $r$ is
\[
\net := \bigotimes_{u\in V\cup W}
\left(
\mem_u \otimes
\bigotimes_{j=1}^{\deg(u)} \psend_{u,j}\ \otimes
\bigotimes_{v\in V \cup W  } \send_{u\to v}\ \otimes
\bigotimes_{v\in V \cup W } \receive_{u\gets v}
\right)
\ \otimes\ 
\bigotimes_{t=1}^{\mu_r} \outbox_t.  
\]

We write $\ket{\phi}_{\mathrm{net}}$ to denote an arbitrary state of this
register.
We maintain the invariant that, at the beginning of each simulated round, all
$\psend_{u,j}$, $\send_{u\to v}$, $\receive_{u\leftarrow v}$, and $\outbox_t$ registers are in the
vacuum state $\perp$.

\paragraph{Preprocessing Unitary. }
Let $\pre$ denote the $X$-independent %
computation of round $r$:
For each node $u$, it updates $u$'s local memory $\mem_u$ by moving the contents of $\receive_{u\leftarrow v}$ for each node $v$, and restores $\receive_{u\leftarrow v}$ to $\perp$. \footnote{For simplicity, we omit the port-receive registers that the distributed algorithm conceptually uses.}
For each port $j$, node $u$ writes the messages that it intends to
send on its local port $j$ into the corresponding port-send registers
$\psend_{u,j}$. 
We interpret the content of $\psend_{u,j}$ after $\pre$ as the message that $u$ sends on port $j$ in this round.
Although $\pre$ is a single unitary acting on all nodes' registers, it is obtained by
grouping together the local operations that each node would perform in a given round of the distributed algorithm. 

\paragraph{Push Unitary.}
For each round $r$, 
we define an $X$-independent \emph{push} unitary $\push$ that coherently 
\emph{transfers the contents} of all non-vacuum port-send registers into the outbox.
More precisely, $\push$ maps each non-vacuum $\psend_{v_i,j}$ into some empty
outbox entry, writes the corresponding indices $(i,j)$ into $(I_t,J_t)$,
and sets the original $\psend_{v_i,j}$ to $\perp$ (unused outbox entries remain $\perp$).
$\push$ can be implemented by a reversible sorting or compaction circuit; e.g., see \cite{asharov}.

\paragraph{Lookup Unitary. }
After applying $\pre$ and $\push$, all messages that are to be routed in round $r$
are represented by the outbox entries where   $\outbox_t:=(I_t,J_t,\Msg_t) = (i, j, m)$ indicates that message $m$ needs to be routed from node $v_i$ via port $j$ in round $r$ ($\outbox_t=\perp$ denoting an empty entry). 
The crucial step in the simulation is to move a message stored in an
outbox entry into the correct edge-level register $\send_{v_i\to u}$, where $u$ is the endpoint of
port $j$ of node $v_i$. Note that this endpoint is not known to the algorithm in advance and may
depend on the value of $X_i$.

Conceptually, we need to query the value of $P_{i,k}$, where $k =  \pi_C(i,j)$.
The query returns an answer bit $b$, where $b = 1$ if $k = X_i$ and $b = 0$ otherwise. 
Then we invoke a (classical) \emph{target function $\target$} to determine the neighbor
$u = \target(i,j,b)$, where
$$
\target(i,j,b) =
\begin{cases}
w_i & \text{if } b = 1,\\
v_{k} & \text{if } b = 0,
\end{cases}
$$ 
Note that $\target$ is a classical function determined by the instance
$G$.

Formally, we introduce three distinguished registers, denoted $I$, $J$, and $B$, which the lookup unitary acts on in our simulation. 

\begin{itemize}
  \item \textbf{Clique Register $I$.}
  The register $I$ is a $\Theta\lt( \log n \rt)$-qubit register with
  computational basis states $\{\ket{i}_I : i \in [n]\}$.  The value $i$
  encodes the identity of the clique node $v_i \in V$ whose value $X_i$ (corresponds to the ID of its matched neighbor in the hidden perfect matching) is being queried.  The simulator may prepare arbitrary superpositions
  $\sum_i \alpha_i \ket{i}_I$ (entangled with other registers), subject only
  to the constraints of the query algorithm.

  \item \textbf{Port register $J$.}
  The register $J$ is an $\Theta\lt( \log n \rt)$-qubit register whose computational basis
  states $\{\ket{j}_J\}$ encode a single local port index $j$ of the currently
  selected clique node $v_i$ (for example $j \in [\deg(v_i)]$).  
  Thus, a basis state $\ket{i}_I \ket{j}_J$ represents $j$-th port of node $v_i$. 

  \item \textbf{Answer qubit $B$.}
  The register $B$ is a single qubit, initialised to $\ket{0}_B$ at the beginning of each query to $O_P$. The lookup unitary acts on
  $(I,J,B)$ by flipping $B$ if and only if $P_{i,k}=1$ where $k = \pi_C(i,j).$ 
  
\end{itemize}

The lookup operator is a unitary $\oracle_X$ which acts only on the registers
$I,J,B$ and as the identity on the joint network register $\net$:
$$
  \oracle_X:
  \ket{i}_I \ket{j}_J \ket{b}_B \otimes \ket{\phi}_{\mathrm{net}}
  \longmapsto
  \ket{i}_I \ket{j}_J
  \ket{b \oplus P_{i, \pi_C(i,j)}}_B
  \otimes
  \ket{\phi}_{\mathrm{net}}
  $$
extended by linearity.
Note that the simulator can implement $O_X$ using a single query to $O_p$.

\paragraph{Load/Unload unitaries.}
For each $t\in[\mu_r]$, we define an $X$-independent \emph{load} unitary $L_t$ that swaps the index
pair $(I,J)$ with the index pair $(I_t,J_t)$ stored in the entry of $\outbox_t$, and acts as the
identity on all other registers. 
Note that $L_t$ is fixed in advance and hence is $X$-independent.

\paragraph{Dispatch unitary.}
For each $t\in[\mu_r]$, we define an $X$-independent controlled unitary $D_t$
that \emph{dispatches} the content of the register $\Msg_t$ of $\outbox_t$ into an
edge-level send register, \emph{conditioned on} the registers $I,J,B$.
In detail, for each triple $(i,j,b)$ and each outbox index $t$, let $D^{(i,j,b)}_t$
denote the unitary on $\net$ that swaps the contents of the outbox message register
$\Msg_t$ with the edge-level register $\send_{v_i\to \target(i,j,b)}$ and acts as the
identity on all other registers:
\[
D^{(i,j,b)}_t = \SWAP_{\Msg_t,\ \send_{v_i\to \target(i,j,b)}}\ \otimes\ I_{\text{rest}}.
\]
We then define
\[
D_t = \sum_{i,j,b} \ket{i}_I\ket{j}_J\ket{b}_B\!\bra{i}_I\bra{j}_J\bra{b}_B\ \otimes\ D^{(i,j,b)}_t.
\]
Thus $D_t$ is block-diagonal with respect to the computational basis of $IJB$.
Since each block $D^{(i,j,b)}_t$ is unitary, $D_t$ is unitary.
\footnote{Here $(i,j,b)$ are classical indices; $D_t^{(i,j,b)}$ acts only on $\net$.
The controlled unitary $D_t$ uses the quantum registers $I,J,B$ to select the appropriate block.
}

\paragraph{Routing a single outbox entry.}
To simplify notation, we define $f(i,j) = P_{i, \pi_C(i,j)}$. That is, $f$ is the indicator function that is set to $1$ if and only if the node $v_i$ is matched to node $v_k$ where $k$ is the ID at the endpoint of port $j$ of $v_i$ as assigned by $\pi_C$ (otherwise,  $f(i,j)=0$). In other words, $f(i,j) =1$ if and only if  $X_i = \pi_C(i,j)$. 
To route $\outbox_t=(I_t,J_t,\Msg_t)$, the simulator applies
$$L_t\circ O_X\circ D_t\circ O_X\circ L_t$$ (with $B$ initialised to $\ket{0}$).
If $\outbox_t=\perp$, then $\Msg_t=\ket{\perp}$ and the swap in $D_t$ has no effect. Otherwise,
for $\outbox_t=(I_t,J_t,\Msg_t)=(i,j,m)$, after the above sequence the message $m$ is placed
in $\send_{v_i\to\target(i,j,f({i,j}))}$ and the answer qubit $B$ is returned to $\ket{0}_B$ and can be
reused. 

\paragraph{Edge-delivery unitary.}
The unitary $\route$ is $X$-independent. It acts only on the edge-level
send/receive registers and as the identity on all other registers: 
$$
  \route
  = \bigotimes_{\{u, v\} \in V(G) \times V(G), u \neq v}
    \SWAP_{\send_{u\to v},\,\receive_{v\leftarrow u}}
    \otimes I_{\text{rest}},
$$ 
where $I_{\text{rest}}$ denotes the identity on all registers other than the
edge-level send/receive registers. 

Informally, $\route$ \emph{delivers} each dispatched message: every non-vacuum
edge-send register $\send_{u\to v}$ is swapped into the corresponding
edge-receive register $\receive_{v\leftarrow u}$.

In Appendix~\ref{app:example}, we give a concrete example of how the above unitaries are applied. 

\paragraph{Query complexity.}
All unitaries except the search operator $O_X$ are $X$-independent.
Moreover, for each $t\in[\mu_r]$ the simulator invokes $O_X$ twice (once to
compute and once to uncompute the answer qubit) within the sequence
$$L_t\circ O_X\circ D_t\circ O_X\circ L_t.$$
Hence the simulation of round $r$ uses at most $2\mu_r$ queries. (Recall that messages sent by the nodes in $W$ can be simulated trivially without queries.)
Summing over all rounds, the total query complexity is
$O\left(\sum_r \mu_r\right)=O(\mu)$. 

\paragraph{Solution to $\desc$.}
To complete the proof of Lemma~\ref{lem:qsim}, it remains to show that any feasible solution to $\matching$ in $G$ induces a valid solution to the $\desc$ problem.  
By construction, $X_i = k$ implies that
$P_{i,k} = P_{k,i}=1$ and $P_{i,k'} = P_{k',i}=0$ for $k'\neq k$. 
Therefore, a solution to $\matching$  reveals the permutation matrix $P$, and the single-bit descriptor required by $\desc$ can be obtained directly from $P$, yielding a valid solution to $\desc$.
\end{proof}

\appendix
\section{Example: a clique node with four ports.} \label{app:example}
Suppose a fixed clique node $v_i$ has four ports $1,2,3,4$. In one round, the
local quantum operation at $v_i$ prepares the superposition
\[
  \frac{1}{\sqrt{2}}\,
  \bigl|\text{send } m_1 \text{ on port } 1,\ m_2 \text{ on port } 2\bigr\rangle
  \;+\;
  \frac{1}{\sqrt{2}}\,
  \bigl|\text{send } m_3 \text{ on port } 3,\ m_4 \text{ on port } 4\bigr\rangle.
\]
Upon measuring this state in the computational basis, the algorithm sends
$m_1$ and $m_2$ on ports $1$ and $2$, respectively, with probability $1/2$,
and $m_3$ and $m_4$ on ports $3$ and $4$, respectively, with probability $1/2$.

After applying the preprocessing unitary $\pre$, the port-send registers
at $v_i$ are in the corresponding superposition:
\[
  \ket{\psi_0}:=\frac{1}{\sqrt{2}}\bigl(\ket{\phi_{12}}+\ket{\phi_{34}}\bigr)\otimes\ket{\mathrm{rest}}.
\]

where
\begin{align*}
  \ket{\phi_{12}}
  &=
  \ket{m_1}_{\psend_{v_i,1}}
  \ket{m_2}_{\psend_{v_i,2}}
  \ket{\vac}_{\psend_{v_i,3}}
  \ket{\vac}_{\psend_{v_i,4}}
  \otimes
  \bigotimes_{u \in  V \cup W} \ket{\vac}_{\send_{v_i\to u}},\\
  \ket{\phi_{34}}
  &=
  \ket{\vac}_{\psend_{v_i,1}}
  \ket{\vac}_{\psend_{v_i,2}}
  \ket{m_3}_{\psend_{v_i,3}}
  \ket{m_4}_{\psend_{v_i,4}}
  \otimes
  \bigotimes_{u \in V \cup W } \ket{\vac}_{\send_{v_i\to u}},
\end{align*}
and $\ket{\mathrm{rest}}$ contains all other registers (including other nodes’
memory and ports).  

After applying the push unitary $\push$, the two non-vacuum port-send registers
in each branch are coherently transferred into the two outbox entries (and all
$\psend_{v_i,j}$ are reset to $\vac$).  Thus we obtain
\[
  \ket{\psi_1}
  \;:=\;
  \frac{1}{\sqrt{2}}
  \Bigl(
    \ket{\phi'_{12}} + \ket{\phi'_{34}}
  \Bigr)\otimes \ket{\mathrm{rest}},
\]
where $\ket{\mathrm{rest}}$ again contains all registers not shown explicitly,
and
\begin{align*}
  \ket{\phi'_{12}}
  &=
  \ket{i}_{I_1}\ket{1}_{J_1}\ket{m_1}_{\Msg_1}
  \ \otimes\
  \ket{i}_{I_2}\ket{2}_{J_2}\ket{m_2}_{\Msg_2}
  \ \otimes\
  \bigotimes_{j=1}^4 \ket{\vac}_{\psend_{v_i,j}}
  \ \otimes\
  \bigotimes_{u \in V \cup W} \ket{\vac}_{\send_{v_i\to u}},\\
  \ket{\phi'_{34}}
  &=
  \ket{i}_{I_1}\ket{3}_{J_1}\ket{m_3}_{\Msg_1}
  \ \otimes\
  \ket{i}_{I_2}\ket{4}_{J_2}\ket{m_4}_{\Msg_2}
  \ \otimes\
  \bigotimes_{j=1}^4 \ket{\vac}_{\psend_{v_i,j}}
  \ \otimes\
  \bigotimes_{u \in V \cup W} \ket{\vac}_{\send_{v_i\to u}}.
\end{align*}
(Any unused outbox entries, if $\mu_r>2$, remain in the vacuum state.)

\paragraph{Routing outbox entry $t=1$.}
The simulator applies the sequence $L_1\circ \oracle_X\circ D_1\circ \oracle_X\circ L_1$,
with $B$ initialised to $\ket{0}_B$.  Consider first the branch $\ket{\phi'_{12}}$,
in which $(I_1,J_1,\Msg_1)=(i,1,m_1)$.  After loading by $L_1$, the workspace
registers contain $(I,J)=(i,1)$; then $\oracle_X$ computes the predicate
$b = f(i,1) =  P_{i, \pi_C(i,1)}$ into $B$; then $D_1$ swaps $\Msg_1$ into the edge-level
register $\send_{v_i\to\target(i,1,b)}$; finally the second application of
$\oracle_X$ uncomputes $B$ and the final $L_1$ restores $(I_1,J_1)$ in the outbox.
Since $\send_{v_i\to u}$ registers start in $\vac$ for all $u \in V\cup W $, this has the net effect
\[
  \ket{m_1}_{\Msg_1}\bigotimes_{u\in V \cup W}\ket{\vac}_{\send_{v_i\to u}}
  \;\longmapsto\;
  \ket{\vac}_{\Msg_1}\,
  \ket{m_1}_{\send_{v_i\to x}}
  \bigotimes_{u \in V \cup W \setminus \{x\} }
  \ket{\vac}_{\send_{v_i\to u}},
\]
where $x = \target(i,1,f(i,1))$,  
while returning $B$ to $\ket{0}_B$ and leaving all other registers unchanged.
In the branch $\ket{\phi'_{34}}$, the same sequence routes $\Msg_1=m_3$ using
the indices $(i,3)$, placing $m_3$ in $\send_{v_i\to \target(i,3,f(i,3))}$ and
resetting $\Msg_1$ to $\vac$.

\paragraph{Routing outbox entry $t=2$.}
Next the simulator applies $L_2\circ \oracle_X\circ D_2\circ \oracle_X\circ L_2$.
In the branch $\ket{\phi'_{12}}$, the outbox entry $\outbox_2$ contains $(i,2,m_2)$,
so the sequence places $m_2$ into $\send_{v_i\to\target(i,2,f(i,2))}$ and resets
$\Msg_2$ to $\vac$ (again returning $B$ to $\ket{0}_B$).  In the branch $\ket{\phi'_{34}}$,
the same sequence places $m_4$ into $\send_{v_i\to\target(i,4,f(i,4))}$ and resets $\Msg_2$.

After routing both outbox entries, we therefore obtain a state $\ket{\psi_2}$ in which
all port-send and outbox message registers are in $\vac$, and the edge-send registers
encode the intended outgoing messages:
\[
  \ket{\psi_2}
  \;=\;
  \frac{1}{\sqrt{2}}
  \Bigl(
    \ket{\Phi_{12}^{\mathrm{send}}} + \ket{\Phi_{34}^{\mathrm{send}}}
  \Bigr)\otimes \ket{\mathrm{rest}},
\]
where $\ket{\Phi_{12}^{\mathrm{send}}}$ has $m_1$ in $\send_{v_i\to\target(i,1,f(i,1))}$
and $m_2$ in $\send_{v_i\to\target(i,2,f(i,2))}$ (all other $\send$ registers vacuum),
and $\ket{\Phi_{34}^{\mathrm{send}}}$ has $m_3$ in $\send_{v_i\to\target(i,3,f(i,3))}$
and $m_4$ in $\send_{v_i\to\target(i,4,f(i,4))}$.

\paragraph{Edge delivery.}
Finally, applying the edge-delivery unitary $\route$ swaps each populated
edge-send register into the corresponding edge-receive register at the
neighbor, i.e., $\send_{u\to v}$ is swapped into $\receive_{v\leftarrow u}$ for
every $\{u,v\}\in E$ (in both directions).  Thus the messages dispatched from
$v_i$ are delivered to the appropriate neighbors. 

\paragraph{Query count in this example.}
Since each routed outbox entry uses exactly two invocations of $\oracle_X$,
and here $\mu_r=2$, this round uses exactly $4$ queries.

\bibliographystyle{alpha}
\bibliography{refs}
\end{document}